\documentclass[american,english,acmsmall,screen]{acmart}
\usepackage[T1]{fontenc}
\usepackage[utf8x]{inputenc}
\usepackage{array}
\usepackage{graphicx}

\makeatletter

\providecommand{\tabularnewline}{\\}

\providecommand*{\code}[1]{\texttt{#1}}

\setcopyright{rightsretained}
\acmPrice{}
\acmDOI{10.1145/3434339}
\acmYear{2021}
\copyrightyear{2021}
\acmSubmissionID{popl21main-p660-p}
\acmJournal{PACMPL}
\acmVolume{5}
\acmNumber{POPL}
\acmArticle{58}
\acmMonth{1}


\makeatother

\usepackage{babel}
\begin{document}
\author{Jules Jacobs}
\affiliation{
   \institution{Radboud University and Delft University of Technology}
   \country{The Netherlands}
}
\email{julesjacobs@gmail.com}
\title{Paradoxes of Probabilistic Programming}
\subtitle{And How to Condition on Events of Measure Zero with Infinitesimal
Probabilities}
\begin{abstract}
Probabilistic programming languages allow programmers to write down
conditional probability distributions that represent statistical and
machine learning models as programs that use observe statements. These
programs are run by accumulating likelihood at each observe statement,
and using the likelihood to steer random choices and weigh results
with inference algorithms such as importance sampling or MCMC. We
argue that naive likelihood accumulation does not give desirable semantics
and leads to paradoxes when an observe statement is used to condition
on a measure-zero event, particularly when the observe statement is
executed conditionally on random data. We show that the paradoxes
disappear if we explicitly model measure-zero events as a limit of
positive measure events, and that we can execute these type of probabilistic
programs by accumulating \emph{infinitesimal probabilities} rather
than probability densities. Our extension improves probabilistic programming
languages as an executable notation for probability distributions
by making it more well-behaved and more expressive\emph{, }by allowing
the programmer to be explicit about\emph{ }which limit is intended
when conditioning on an event of measure zero.
\end{abstract}
\maketitle

\section{Introduction}

Probabilistic programming languages such as Stan \citep{Carpenter_stan:a},
Church \citep{GoodmanMRBT08}, and Anglican \citep{WoodMM14} allow
programmers to express probabilistic models in statistics and machine
learning in a structured way, and run these models with generic inference
algorithms such as importance sampling, Metropolis-Hastings, SMC,
HMC. At its core, a probabilistic programming language is a notation
for probability distributions that looks much like normal programming
with calls to random number generators, but with an additional \code{observe}
construct.

There are two views on probabilistic programming. The pragmatist says
that probabilistic programs are a convenient way to write down a likelihood
function, and the purist says that probabilistic programs are a notation
for structured \foreignlanguage{american}{probabilistic} models. The
pragmatist interprets an \code{observe} statement as ``soft conditioning'',
or imperatively multiplying the likelihood function by some factor.
The purist interprets an \code{observe} statement as true probabilistic
conditioning in the sense of conditional distributions. The pragmatist
may also want to write a probabilistic program to compute the likelihood
function of a conditional distribution, but the pragmatist is not
surprised that there are non-sensical probabilistic programs that
do not express any sensible statistical model. After all, if one writes
down an arbitrarily likelihood function then it will probably not
correspond to a sensible, structured, non-trivial statistical model.
The pragmatist blames the programmer for writing non-sensical programs,
just as it would have been the fault of the programmer if they had
written down the same likelihood function manually. The purist, on
the other hand, insists that \emph{any} probabilistic program corresponds
to structured statistical model, and that each \code{observe} statement
in a probabilistic program has a probabilistic interpretation whose
composition results in the statistical model. We will show that the
current state is not satisfactory for the purist, and we will show
how to make probabilistic programming languages satisfactory in this
respect.

The difficulties with conditioning in probabilistic programs can be
traced back to a foundational issue in probability theory. When the
event $E$ being conditioned on has nonzero probability, the conditional
distribution $\mathbb{P}(A|E)$ is defined as:
\begin{align*}
\mathbb{P}(A|E) & =\frac{\mathbb{P}(A\cap E)}{\mathbb{P}(E)}
\end{align*}
However, this formula for conditional probability is undefined when
$\mathbb{P}(E)=0$, since then also $\mathbb{P}(A\cap E)=0$ and the
fraction $\mathbb{P}(A|E)=\frac{0}{0}$ is undefined. In probabilistic
programming we often wish to condition on events $E$ with probability
$0$, such as ``$x=3.4$'', where $x$ is a continuous random variable.
There are several methods to condition on measure-zero events. For
continuous distributions that have probability density functions,
we can replace the probabilities in the above formula with probability
densities, which are (usually) nonzero even if $\mathbb{P}(E)$ is
zero. For more complicated situations, we can use the Radon--Nikodym
derivative or disintegration \citep{Chang93conditioningas,ShanR17,DahlqvistK20,ackerman_freer_roy_2017}.

A general method for conditioning on measure-zero events is to define
a sequence of events $E_{\epsilon}$ parameterized by a number $\epsilon>0$
such that $E_{\epsilon}$ in some sense converges to $E$ in the limit
$\epsilon\to0$, but $\mathbb{P}(E_{\epsilon})>0$ for all $\epsilon>0$.
We then \emph{define} the conditional distribution to be the limit
of $\mathbb{P}(A|E_{\epsilon})$:
\begin{align*}
\mathbb{P}(A|E) & =\lim_{\epsilon\to0}\frac{\mathbb{P}(A\cap E_{\epsilon})}{\mathbb{P}(E_{\epsilon})}
\end{align*}

In the book \emph{Probability Theory: The Logic of Science \citep{jaynes03},}
E.T. Jaynes explains that conditioning on measure-zero events is inherently
ambiguous, because it depends not just on $E$ but also on the limiting
operation $E_{\epsilon}$ we choose:
\begin{quote}
Yet although the sequences $\{A_{\epsilon}\}$ and $\{B_{\epsilon}\}$
tend to the same limit ``$y=0$'', the conditional densities {[}$\mathbb{P}(x|A_{\epsilon})$
and $\mathbb{P}(x|B_{\epsilon})${]} tend to different limits. As
we see from this, merely to specify ``$y=0$'' without any qualifications
is ambiguous; it tells us to pass to a measure-zero limit, but does
not tell us which of any number of limits is intended. {[}...{]} Whenever
we have a probability density on one space and we wish to generate
from it one on a subspace of measure zero, the only safe procedure
is to pass to an explicitly defined limit by a process like {[}$A_{\epsilon}$
and $B_{\epsilon}${]}. In general, the final result will and must
depend on which limiting operation was specified. This is extremely
counter-intuitive at first hearing; yet it becomes obvious when the
reason for it is understood.
\end{quote}
The other methods implicitly make the choice $E_{\epsilon}$ for us.
Conditioning on events of measure-zero using those methods can lead
to paradoxes such as the Borel-Komolgorov paradox, even in the simplest
case when probability density functions exist. Paradoxes occur because
seemingly unimportant restatements of the problem, such as using a
different parameterization for the variables, can affect the choice
of $E_{\epsilon}$ that those methods make, and thus change the value
of the limit. Consider the following probabilistic program:
\begin{verbatim}

  h = rand(Normal(1.7, 0.5))
  if rand(Bernoulli(0.5))
     observe(Normal(h, 0.1), 2.0)
  end

\end{verbatim}
We first sample a value (say, a person's height) from a prior normally
distributed around 1.7 meters and then with probability 0.5 we observe
a measurement normally distributed around the height to be 2.0. We
ran this program in Anglican with importance sampling, and obtained
the following expectation values for $h$: 1.812 1.814 1.823 1.813
1.806 (10000 samples each). Suppose that we had measured the height
in centimeters instead of meters:
\begin{verbatim}

  h = rand(Normal(170, 50))
  if rand(Bernoulli(0.5))
     observe(Normal(h, 10), 200)
  end

\end{verbatim}
We might naively expect this program to produce roughly the same output
as the previous program, but multiplied by a factor of 100 to account
for the conversion of meters to centimeters. Instead, we get 170.1
170.4 171.5 170.2 169.4. This \foreignlanguage{american}{behavior}
happens because even though the units of the program appear to be
correct, the calculations that importance sampling does to estimate
the expectation value involve arithmetic with inconsistent units (in
this case, adding a quantity with units $m^{-1}$ to a quantity with
neutral units). The issue is not particular to Anglican or importance
sampling, but due to the interaction of stochastic branching with
way the likelihood is calculated with probability densities; other
algorithms \citep{paige-nips-2014,tolpin-ecml-2015} have the same
\foreignlanguage{american}{behavior}. In fact, formal semantics based
on likelihood accumulation, such as the commutative semantics \citep{10.1007/978-3-662-54434-1_32}
and the semantics based on on Quasi-Borel spaces \citep{HeunenKSY17},
also perform arithmetic with inconsistent units for this example.
Lexical likelihood weighting \citep{pmlr-v80-wu18f} does give the
right answer for this example\footnote{Many thanks to Alex Lew for pointing this out.},
but still exhibits unit anomalies for other examples described in
Section \ref{sec:Three-types-of}.

Unit errors in a programming language's implementation or semantics
may seem like a very serious issue, but we do not believe that this
is a show-stopper in practice, because practitioners can always take
the pragmatist view and avoid writing such programs. Although we consider
this to be an important foundational issue, it does not invalidate
existing work on probabilistic programming.

It is known that conditionals can be problematic. Some inference algorithms,
like SMC, will make assumptions that exclude \code{observe} inside
conditionals. For example, \citep{probprogfrankwood} mentions the
following when describing SMC:
\begin{quote}
Each breakpoint needs to occur at an expression that is evaluated
in every execution of a program. In particular, this means that breakpoints
should not be associated with expressions inside branches of if expressions.
{[}...{]} An alternative design, which is often used in practice,
is to simply break at every observe and assert that each sample has
halted at the same point at run time.
\end{quote}
If the design is used where breakpoints happen at every observe, then
the assertion that breakpoints should not be associated with expressions
inside branches of if expressions will disallow using SMC with programs
that have observes inside conditionals. Languages such as Stan, that
do not have or do not allow stochastic branching, also do not suffer
from the preceding example. In section \ref{sec:Three-types-of} we
will show that the problem is not limited to conditionals; there are
programs that do not have conditionals but nevertheless have paradoxical
\foreignlanguage{american}{behavior}. Furthermore, we show that the
standard method of likelihood accumulation for implementing probabilistic
programming languages can sometimes obtain an answer that disagrees
with the purist's exact value for $\mathbb{P}(A|E)$ \textbf{even
if $\mathbb{P}(E)$ is nonzero}, due to a confusion between probabilities
and probability densities.

We identify three types paradoxes that affect probabilistic programming
languages that allow dynamically conditioning on events of measure-zero.
These paradoxes are based on the idea that it should not matter which
parameter scale we use for variables. It shouldn't matter whether
we use meters or centimeters to measure height, but it also shouldn't
matter whether we use energy density or decibels to measure sound
intensity. The change from centimeters to meters involves a linear
parameter transformation by $cm=0.01m$, whereas the change from energy
density to decibels involves a nonlinear parameter transformation
$\text{decibels}=\log(\text{energy density})$. We give several example
programs that show that the output of a probabilistic program can
depend on the parameter scale used when we condition on events of
measure zero.

Following Jaynes' advice, we extend the language with notation for
explicitly choosing \emph{which} limit $E_{\epsilon}$ we mean in
an \code{observe} statement. We give an implementation of likelihood
accumulation using \emph{infinitesimal probabilities} instead of probability
densities, and show that this does not suffer from the three types
of paradoxes. Infinitesimal probabilities give meaning to conditioning
on measure-zero events in terms of a limit of events of strictly positive
measure. Since events of strictly positive measure are unproblematic,
paradoxes can no longer occur.

Furthermore, we add explicit language support for parameter transformations.
This is only soundly possible due to the introduction of infinitesimal
probabilities. We show that introducing a parameter transformation
in an \code{observe} statement does not change the \foreignlanguage{american}{behavior}
of the probabilistic program. That is, we show that in our language,
\code{observe(D,I)} has the same behavior as \code{observe(D',I')}
where \code{D',I'} is \code{D,I} in a different parameter scale.

\bigskip{}
Our contributions are the following.
\begin{itemize}
\item We identify a problem with existing probabilistic programming languages,
in which likelihood accumulation with probability densities can result
in three different types of paradoxes when conditioning on a measure-zero
event. The three paradoxes violate the principle that the output of
a program should not depend on the parameter scale used (Section \ref{sec:Three-types-of}).
\item We analyze the event that probabilistic programs with observe statements
condition on, taking the paradox-free discrete case as a guide, in
order to determine what \code{observe} ought to mean in the continuous
case (Section \ref{sec:ON-THE-EVENT}).
\item We propose a change to probabilistic programming languages to avoid
the paradoxes of the continuous measure-zero case, by changing the
\code{observe} construct to condition on measure-zero events $E$
as an explicit limit $\epsilon\to0$ of $E_{\epsilon}$ (Sections
\ref{sec:AVOIDING-EVENTS-OF} and \ref{sec:USING-INFINITESIMAL-NUMBERS}),
and
\begin{itemize}
\item a method for computing the limit by accumulating \emph{infinitesimal
probabilities} instead of probability densities, which we use to implement
the adjusted \code{observe} construct,
\item a theorem that shows that infinitesimal probabilities correctly compute
the limit of $E_{\epsilon}$, ensuring that programs that use observe
on measure-zero events are paradox free,
\item a translation from the existing \code{observe} construct to our new
\code{observe} construct, which gives the same output if the original
program was non-paradoxical,
\item language support for parameter transformations, which we use to show
that the meaning of programs in our language is stable under parameter
transformations,
\item an implementation of our language as an embedded DSL in Julia \citep{measurezeroartifact}
(Section \ref{sec:IMPLEMENTATION-IN-JULIA}).
\end{itemize}
\end{itemize}

\section{On the Event that Observe Conditions On\label{sec:ON-THE-EVENT}}

Different probabilistic programming languages have different variants
of the \code{observe} statement. Perhaps it's simplest variant, \code{observe(b)}
takes a boolean \code{b} and conditions on that boolean being true.
For instance, if we throw two dice and want to condition on the sum
of the dice being 8, we can use this probabilistic program, in pseudocode:\\

\begin{verbatim}
  function twoDice()
     x = rand(DiscreteUniform(1,6))
     y = rand(DiscreteUniform(1,6))
     observe(x + y == 8)
     return x
  end

\end{verbatim}
The program \code{twoDice} represents the conditional distribution
$\mathbb{P}(x|x+y=8)$ where $x$ and $y$ are uniformly distributed
numbers from $1$ to $6$. We wrap the program in a function and use
the return value to specify the variable \code{x} whose distribution
we are interested in. Anglican has a \code{defquery} construct analogous
to the function definition that we use here.

Probabilistic programming languages allow us to sample from the distribution
specified by the probabilistic program and compute expectation values.
The simplest method to implement \code{observe} is \emph{rejection
sampling} \citep{vonNeumann1951,GoodmanMRBT08}: we start a \emph{trial}
by running the program from the beginning, drawing random samples
with \code{rand}, and upon encountering \code{observe(x + y == 8)}
we test the condition, and if the condition is not satisfied we reject
the current trial and restart the program from the beginning hoping
for better luck next time. If all observes in a trial are satisfied,
then we reach the return statement and obtain a sample for \code{x}.
We estimate expectation values by averaging multiple samples.

What makes probabilistic programs such an expressive notation for
probability distributions is that we have access to use the full power
of a programming language, such as its control flow and higher order
functions \citep{HeunenKSY17}. The following example generates two
random dice throws \code{x} and \code{y}, and a random boolean \code{b},
and uses an observe statement to condition on the sum of the dice
throws being \code{8} if \code{b = true}, with control flow:
\begin{verbatim}

  x = rand(DiscreteUniform(1,6))
  y = rand(DiscreteUniform(1,6))
  b = rand(Bernoulli(0.5))
  if b
     observe(x + y == 8)
  end
  return x

\end{verbatim}
This code expresses the conditional probability distribution $\mathbb{P}(x|E)$
where $x,y,b$ are distributed according to the given distributions,
and $E$ is the event $(b=true\wedge x+y=8)\vee(b=false).$ That is,
a trial is successful if $x+y=8$ or if $b=false$.

In general, a probabilistic program conditions on the event that the
tests of all \code{observe} statements that are executed succeed.
A bit more formally, we have an underlying probability space $\Omega$
and we think of an element $\omega\in\Omega$ as the ``random seed''
that determines the outcome of all \code{rand} calls (it is sufficient
to take $\Omega=\mathbb{R}$; a real number contains an infinite amount
of information, sufficient to determine the outcome of an arbitrary
number of \code{rand} calls, even if those calls are sampling from
continuous distributions). The execution trace of the program is completely
determined by the choice $\omega\in\Omega$. For some subset $E\subset\Omega$,
the tests of all the \code{observe} calls that are executed in the
trace succeed. This is the event $E$ that a probabilistic program
conditions on. Rejection sampling gives an intuitive semantics for
the \code{observe} statement:

\begin{table}[H]
\begin{tabular}{|>{\centering}p{0.55\paperwidth}|}
\hline
For a boolean \code{b}, the statement \code{observe(b)} means that
we only continue with the current trial only if \code{b = true}.
If \code{b = false} we reject the current trial.\tabularnewline
\hline
\end{tabular}
\end{table}

Unfortunately, rejection sampling can be highly inefficient when used
to run a probabilistic program. If we use 1000-sided dice instead
of 6-sided dice, the probability that the sum $x+y$ is a particular
fixed value is very small, so most trials will be rejected and it
may take a long time to obtain a successful sample. Probabilistic
programming languages therefore have a construct \code{observe(D,x)}
that means \code{observe(rand(D) == x)}, but can be handled by more
efficient methods such as importance sampling or Markov Chain Monte
Carlo (MCMC). The previous example can be written using this type
of \code{observe} as follows:
\begin{verbatim}

  x = rand(DiscreteUniform(1,6))
  b = rand(Bernoulli(0.5))
  if b
     observe(DiscreteUniform(1,6), 8 - x)
  end
  return x

\end{verbatim}
This relies on the fact that \code{x + y == 8} is equivalent to \code{y == 8 - x}.
The intuitive semantics of \code{observe(D,x)} is as follows:

\begin{table}[H]
\begin{tabular}{|>{\centering}p{0.55\paperwidth}|}
\hline
\centering{}For discrete distributions \code{D}, the statement \code{observe(D,x)}
means that we sample from \code{D} and only continue with the current
trial if the sampled value is equal to \code{x}.\tabularnewline
\hline
\end{tabular}
\end{table}

This variant of \code{observe} can be implemented more efficiently
than rejection sampling. We keep track of the \code{weight} of the
current trial that represents the probability that the trial is still
active (i.e. the probability that it was not yet rejected). An \code{observe(D,x)}
statement will multiply the \code{weight} of the current trial by
the probability \code{P(D,x)} that a sample from \code{D} is equal
to \code{x}:

\begin{table}[H]
\begin{tabular}{|>{\centering}p{0.55\paperwidth}|}
\hline
\centering{}For discrete distributions \code{D}, the statement \code{observe(D,x)}
gets executed as \code{weight {*}= P(D,x)}, where \code{P(D,x)}
is the probability of \code{x} in \code{D}.\tabularnewline
\hline
\end{tabular}
\end{table}

The output of a trial of a probabilistic program is now weighted sample:
a pair of random value \code{x} and a \code{weight}. Weighted samples
can be used to compute expectation values as weighted averages (this
is called \emph{importance sampling} \footnote{More advanced MCMC methods can use the weight to make intelligent
choices for what to return from \code{rand} calls, whereas importance
sampling uses a random number generator for \code{rand} calls. We
focus on importance sampling because this is the simplest method beyond
rejection sampling.}). Estimating an expectation value using importance sampling will
usually converge faster than rejection sampling, because importance
sampling's \code{observe} will deterministically weigh the trial
by the probability \code{P(D,x)} rather than randomly rejecting the
trial with probability \code{1 - P(D,x)}. If \code{P(D,x) = 0.01}
then rejection sampling would reject 99\% of trials, which is obviously
very inefficient. It is important to note that multiplying \code{weight {*}= P(D,x)}
is the optimized \emph{implementation} of \code{observe}, and we
may still semantically think of it as rejecting the trial if \code{sample(D) != x}.

If the distribution \code{D} is a continuous distribution, then the
probability that a sample from \code{D} is equal to any particular
value \code{x} becomes zero, so rejection sampling will reject $100\%$
of trials; it becomes infinitely inefficient. This is not surprising,
because on the probability theory side, the event $E$ that we are
now conditioning on has measure zero. Importance sampling, on the
other hand, continues to work in some cases, provided we replace probabilities
with probability densities:

\begin{table}[H]
\begin{tabular}{|>{\centering}p{0.55\paperwidth}|}
\hline
\centering{}For continuous distributions \code{D}, the statement
\code{observe(D,x)} gets executed as \code{weight {*}= pdf(D,x)},
where \code{pdf(D,x)} is the probability density of \code{x} in
\code{D}.\tabularnewline
\hline
\end{tabular}
\end{table}

For instance, if we want to compute $\mathbb{E}[x|x+y=8]$ where $x$
and $y$ are distributed according to $Normal(2,3)$ distributions,
conditioned on their sum being $8$, we can use the following probabilistic
program:
\begin{verbatim}

  x = rand(Normal(2,3))
  observe(Normal(2,3), 8 - x)
  return x

\end{verbatim}
This allows us to draw (weighted) samples from the distribution $\mathbb{P}(x|x+y=8)$
where $x,y$ are distributed according to $Normal(2,3)$. Unfortunately,
as we shall see in the next section, unlike the discrete case, we
do not in general have a probabilistic interpretation for \code{observe(D,x)}
on continuous distributions \code{D} when control flow is involved,
and we can get paradoxical \foreignlanguage{american}{behavior} even
if control flow is not involved.

\section{Three Types of Paradoxes \label{sec:Three-types-of}}

We identify three types of paradoxes. The first two involve control
flow where we either execute observe on different variables in different
control flow paths, or an altogether different number of observes
in different control flow paths. The third paradox is a variant of
the Borel-Komolgorov paradox and involves non-linear parameter transformations.

\subsection{Paradox of Type 1: Different Variables Observed in Different Control
Flow Paths}

Consider the following probabilistic program:
\begin{verbatim}

  h = rand(Normal(1.7, 0.5))
  w = rand(Normal(70, 30))
  if rand(Bernoulli(0.5))
     observe(Normal(h, 0.1), 2.0)
  else
     observe(Normal(w, 5), 90)
  end
  bmi = w / h^2

\end{verbatim}
We sample a person's height $h$ and weight $w$ from a prior, and
then we observe a measurement of the height or weight depending on
the outcome of a coin flip. Finally, we calculate the BMI, and want
to compute its average. If $h'$ is the measurement sampled from $\mathsf{Normal}(h,0.1)$
and $w'$ is the measurement sampled from $\mathsf{Normal}(w,5)$
and $b$ is the boolean sampled from $\mathsf{Bernoulli}(0.5)$, then
the event that this program conditions on is $(b=\mathsf{true}\wedge h'=2.0)\vee(b=\mathsf{false}\wedge w'=90)$.
This event has measure zero.

Just like the program in the introduction, this program exhibits surprising
behavior when we change $h$ from meters to centimeters: even after
adjusting the formula $\mathit{bmi}=w/(0.01\cdot h)^{2}$ to account
for the change of units, the estimated expectation value for $\mathit{bmi}$
still changes. Why does this happen?

The call to \code{observe(D,x)} is implemented as multiplying the
\code{weight} by the probability density of \code{x} in \code{D}.
Importance sampling runs the program many times, and calculates the
estimate for \code{bmi} as a weighted average. Thus the program above
effectively gets translated as follows by the implementation:
\begin{verbatim}

  weight = 1
  h = rand(Normal(1.7, 0.5))
  w = rand(Normal(70, 30))
  if rand(Bernoulli(0.5))
     weight *= pdf(Normal(h, 0.1), 2.0)
  else
     weight *= pdf(Normal(w, 90), 5)
  end
  bmi = w / h^2

\end{verbatim}
Where $\mathsf{pdf}(\mathsf{Normal}(\mu,\sigma),x)$ is the probability
density function of the normal distribution:
\begin{align*}
\mathsf{pdf}(\mathsf{Normal}(\mu,\sigma),x) & =\frac{1}{\sigma\sqrt{2\pi}}e^{-\frac{1}{2}(\frac{x-\mu}{\sigma})^{2}}
\end{align*}
 Importance sampling runs this program $N$ times, obtaining a sequence
$(\mathit{bmi}_{k},\mathit{weight}_{k})_{k\in\{1,\dots,N\}}$. \\
It estimates $\mathbb{E}[\mathit{bmi}]$ with a weighted average:
\begin{align*}
\mathbb{E}[\mathit{bmi}] & \approx\frac{\sum_{k=1}^{N}(\mathit{weight}_{k})\cdot(\mathit{bmi}_{k})}{\sum_{k=1}^{N}(\mathit{weight}_{k})}
\end{align*}

The problem that causes this estimate to change if we change the units
of \code{h} is that the formula adds quantities with inconsistent
units: some $\mathit{weight}_{k}$ have unit $m^{-1}$ (inverse length)
and some have unit $kg^{-1}$\emph{ }(inverse mass)\emph{.}

It might be surprising that the weights have units at all, but consider
that if we have a probability distribution $D$ over values of unit
$U$, then the probability density function $\mathsf{pdf}(D,x)$ has
units $U^{-1}$. The formula for $\mathsf{pdf}(\mathsf{Normal}(\mu,\sigma),x)$
shows this in the factor of $\frac{1}{\sigma}$ in front of the (unitless)
exponential, which has a unit because $\sigma$ has a unit.

The call \code{pdf(Normal(h, 0.1), 2.0)} has units $m^{-1}$ and
the call \code{pdf(Normal(w, 90), 5)} has units $kg^{-1}$, and thus
the variable \code{weight} has units $m^{-1}$ or $kg^{-1}$ depending
on the outcome of the coin flip. The weighted average estimate for
$\mathbb{E}[\mathit{bmi}]$ adds weights of different runs together,
which means that it adds values of unit $m^{-1}$ to values of unit
$kg^{-1}$ . This manifests itself in the estimate changing depending
on whether we use $m$ or $cm$: computations that do arithmetic with
inconsistent units may give different results depending on the units
used. This calls into question whether this estimate is meaningful,
since the estimate depends on whether we measure a value in $m$ or
$cm$, or in $kg$ or $g$, which arguably should not matter at all.

The reader might now object that conditionally executed \code{observe}
statements are always wrong, and probabilistic programs that use them
should be rejected as erroneous. However, in the discrete case there
are no unit errors, because in that case the weight gets multiplied
by a \emph{probability} rather than a \emph{probability density},
and probabilities are unitless. Furthermore, in the preceding section
we have seen that conditionally executed observe statements have a
rejection sampling interpretation in the discrete case. This gives
the programs a probabilistic meaning in terms of conditional distributions,
even if the discrete observe statements are inside conditionals. The
event $E$ that is being conditioned on involves the boolean conditions
of the control flow. Ideally we would therefore not want to blame
the programmer for using conditionals, but change the implementation
of \code{observe} on continuous variables so that the program is
meaningful in the same way that the analogous program on discrete
variables is meaningful.

\subsection{Paradox of Type 2: Different Number of Observes in Different Control
Flow Paths}

Let us analyze the program from the introduction:
\begin{verbatim}

  h = rand(Normal(1.7, 0.5))
  if rand(Bernoulli(0.5))
     observe(Normal(h, 0.1), 2.0)
  end
  return h

\end{verbatim}
This program exhibits unit anomalies for the same reason: some of
the $\mathit{weight}_{k}$ have units $m^{-1}$ and some have no units,
and adding those leads to the surprising behavior. Rather than taking
this behavior as a given, let us analyze what this program \emph{ought}
to do, if we reason by analogy to the discrete case.

This program has the same structure as the dice program from section
2, the difference being that we now use a normal distribution instead
of a discrete uniform distribution. By analogy to that discrete case,
the event that is being conditioned on is $(b=\mathsf{true}\wedge h'=2.0)\vee(b=\mathsf{false})$,
where $h'$ is the measurement from $\mathsf{Normal}(h,0.1)$.

Surprisingly, this event \emph{does not have measure zero!} The event
$(b=\mathsf{true}\vee h'=2.0)$ has measure zero, but the event $b=\mathsf{false}$
has measure $\frac{1}{2}$, so the entire event has measure $\frac{1}{2}$.
We can therefore unambiguously apply the definition of conditional
probability $\mathbb{P}(A|E)=\frac{\mathbb{P}(A\cap E)}{\mathbb{P}(E)}$.
Our probability space is $\Omega=\mathbb{R}\times\mathbb{R}\times\mathsf{bool}$,
corresponding to $h\sim\mathsf{Normal}(1.7,0.5)$, $h'\sim\mathsf{Normal}(h,0.1)$,
$b\sim\mathsf{Bernoulli}(0.5)$, and $A\subseteq\Omega$ and $E=\{(h,h',b)|(b=\mathsf{true}\wedge h'=2.0)\vee(b=\mathsf{false})\}\subseteq X$.
The posterior $\mathbb{P}(A|E)=\frac{P(A\cap E)}{P(E)}=2\cdot\mathbb{P}(A\cap E)=2\cdot\mathbb{P}(A\cap\{(h,h',b)|b=\mathsf{false}\})$,
so the marginal posterior for $h$ is simply $\mathsf{Normal}(1.7,0.5)$.
That is, the whole if statement with the \code{observe} ought to
have no effect.

We can understand this intuitively in terms of rejection sampling:
if the sampled boolean $b=\mathsf{true}$, then the observe statement
will reject the current trial with probability 1, because the probability
of sampling exactly 2.0 from a normal distribution is zero. Hence
if $b=\mathsf{true}$ then the trial will almost surely get rejected,
whereas if $b=\mathsf{false}$ the trial will not get rejected. The
trials where $b=\mathsf{true}\wedge h'=2.0$ are negligibly rare,
so even though the expectation of $h$ is affected \emph{in those
trials}, they do not contribute to the final expectation value; only
trials with $b=\mathsf{false}$ do.

As an aside: if we added an extra \emph{unconditional} \code{observe(Normal(h, 0.1), 1.9)}\emph{
}to the program, then the whole event will have measure zero, but
nevertheless, trials with $b=\mathsf{false}$ will dominate over trials
with $b=\mathsf{true}$, relatively speaking. In general, the control
flow path with the least number of continuous observes dominates.
If there are multiple control flow paths with minimal number of observes,
but also control flow paths with a larger number of observes, we may
have a paradox of mixed type 1 \& 2.

This reasoning would imply that the if statement and the observe statement
are irrelevant; the program ought to be equivalent to \code{return rand(Normal(1.7, 0.5))}.
If this still seems strange, consider the following discrete analogue:
\begin{verbatim}

  h = rand(Binomial(10000, 0.5))
  if rand(Bernoulli(0.5))
     observe(binomial(10000, 0.9), h)
  end
  return h

\end{verbatim}
That is, we first sample $h$ between 0 and 10000 according to a binomial
distribution, and then with probability $0.5$ we observe that $h$
is equal to a number sampled from another binomial distribution that
gives a number between $0$ and $10000$. Since that binomial distribution
is highly biased toward numbers close to $10000$, we might expect
the average value of $h$ to lie significantly higher than $5000$.
This is not the case. The rejection sampling interpretation tells
us that most of the trials where the coin flipped $\mathsf{true}$,
will be rejected, because the sample from $\mathsf{Binomial}(10000,0.9)$
is almost never equal to $h$. Thus, although \emph{those} samples
have an average significantly above $5000$, almost all of the successful
trials will be trials where the coin flipped $\mathsf{false}$, and
thus the expected value of $h$ will lie very close to $5000$.

Since we know that rejection sampling agrees with importance sampling
in expectation, importance sampling will also compute an estimate
for the expectation value of $h$ that lies very close to $5000$.
The further we increase the number 10000, the stronger this effect
becomes, because the probability that the second sample is equal to
$h$ further decreases. In the continuous case this probability becomes
$0$, so the successful samples will almost surely be from trials
where the coin flipped to $\mathsf{false}$. Therefore the average
value of $h$ in the continuous case should indeed be $170$, unaffected
by the if statement and the observe.

Another way to express this point, is that in the discrete case importance
sampling, rejection sampling, and the exact value given by the conditional
expectation are all in agreement, even if conditionals are involved.
On the other hand, in the continuous case, importance sampling with
probability densities gives a different answer than rejection sampling
and the exact value given by the conditional expectation $\mathbb{E}[h|E]$
(the latter two \emph{are} equal to each other; both $1.7$).

The reader may insist that the semantics of the program is \emph{defined}
to be weight accumulation with probability densities, that is, the
semantics of the program is \emph{defined} to correspond to
\begin{verbatim}

  weight = 1
  h = rand(Normal(1.7, 0.5))
  if rand(Bernoulli(0.5))
     weight *= pdf(Normal(h, 0.1), 2.0)
  end
  return h

\end{verbatim}
We can only appeal to external principles to argue against this, such
as unit consistency, analogy with the discrete case, the probabilistic
interpretation of observe, and the rejection sampling interpretation
of observe, but the reader may choose to lay those principles aside
and take this \emph{implementation} of observe as the \emph{semantics}
of observe. We do hope to eventually convince this reader that a \emph{different}
implementation of observe that does abide by these principles, could
be interesting. Although our semantics will differ from the standard
one, it will agree with lexicographic likelihood weighting\citep{pmlr-v80-wu18f}
for this example, which does not exhibit this particular paradox.

\subsection{Paradox of Type 3: Non-Linear Parameter Transformations}

Consider the problem of conditioning on $x=y$ given $x\sim\mathsf{Normal}(10,5)$
and $y\sim\mathsf{Normal}(15,5)$, and computing the expectation $\mathbb{E}[\exp(x)]$.
Written as a probabilistic program,
\begin{verbatim}

  x = rand(Normal(10,5))
  observe(Normal(15,5),x)
  return exp(x)

\end{verbatim}
In a physical situation, $x,y$ might be values measured in decibels
and $\exp(x),\exp(y)$ may be (relative) energy density. We could
change parameters to $a=\exp(x)$ and $b=\exp(y)$. Then $a\sim\mathsf{LogNormal}(10,5)$
and $b\sim\mathsf{LogNormal}(15,5)$. Since the event $x=y$ is the
same as $\exp(x)=\exp(y)$, we might naively expect the program to
be equivalent to:
\begin{verbatim}

  a = rand(LogNormal(10,5))
  observe(LogNormal(15,5),a)
  return a

\end{verbatim}
This is not the case. The two programs give different expectation
values. Compared to type 1 \& 2 paradoxes, this type 3 paradox shows
that the subtlety is not restricted to programs that have control
flow or to distributions that are not continuous; the normal and lognormal
distributions are perfectly smooth.

This paradox is closely related to the Borel-Komolgorov paradox. Another
variant of the original Borel-Komolgorov paradox is directly expressible
in Hakaru \citep{ShanR17}, but not in Anglican or Stan. Hakaru allows
the programmer to condition a measure-zero condition $f(x,y)=0$ such
as $x+y-8=0$ directly without having to manually invert the relationship
to $y=8-x$, and performs symbolic manipulation to do exact Bayesian
inference. Hakaru allows a single such observe at the very end of
a program, which allows it to sidestep the previous paradoxes related
to control flow. The semantics of the single observe is defined by
disintegration, which means that the semantics of a Hakaru program
depends on the form of $f$. That is, if we take another function
$g$ with the same solution set $g(x,y)=0$ as $f$, the output may
change. The programmer can use this mechanism to control which event
they want to condition on. Our version of the paradox shows that the
subtlety of conditioning on measure-zero events is not restricted
to programs that use that type of disintegration.

\section{Avoiding Events of Measure Zero with Intervals\label{sec:AVOIDING-EVENTS-OF}}

Unit anomalies cannot occur with discrete distributions, because in
the discrete case we only deal with probabilities and not with probability
densities. Recall that for discrete probability distributions \code{D},
an observe statement \code{observe(D,x)} gets executed as \code{weight {*}= P(D,x)}
where \code{P(D,x)} is the probability of $x$ in the distribution
$D$. Probabilities have no units, so the \code{weight} variable
stays unitless and the weighted average is always unit correct if
the probabilistic program is unit correct, even if \code{observe}
statements get executed conditionally. Furthermore, in the discrete
case we have a probabilistic and rejection sampling interpretation
of observe, and we may view weight accumulation as an optimization
to compute the same expectation values as rejection sampling, but
more efficiently. We wish to extend these good properties to the continuous
case.

The reason that the discrete case causes no trouble is not due to
\code{D} being discrete per se. The reason it causes no trouble is
that \code{P(D,x)} is a probability rather than a probability density.
In the continuous case the probability that \code{rand(D) == x} is
zero, and that's why it was necessary to use probability densities.
However, even in the continuous case, the probability that a sample
from \code{D} lies in some \emph{interval} is generally nonzero.
We shall therefore change the observe statement to \code{observe(D,I)}
where \code{I} is an interval, which conditions on the event $\mathsf{rand}(D)\in I$.
In the discrete case we can allow \code{I} to be a singleton set,
but in the continuous case we insist that \code{I} is an interval
of nonzero width.

We have the following rejection sampling interpretation for \code{observe(D,I)}:

\begin{table}[H]
\begin{tabular}{|>{\centering}p{0.55\paperwidth}|}
\hline
\centering{}For continuous or discrete distributions \code{D}, the
statement \code{observe(D,I)} means that we sample from \code{D}
and only continue with the current trial if the sampled value lies
in \code{I}.\tabularnewline
\hline
\end{tabular}
\end{table}

And the following operational semantics for \code{observe(D,I)}:

\begin{table}[H]
\begin{tabular}{|>{\centering}p{0.55\paperwidth}|}
\hline
\centering{}For continuous or discrete distributions \code{D}, the
statement \code{observe(D,I)} gets executed as \code{weight {*}= P(D,I)}
where \code{P(D,I)} is the probability that a value sampled from
\code{D} lies in \code{I}.\tabularnewline
\hline
\end{tabular}
\end{table}

Let $I=[a,b]=\{x\in\mathbb{R}\,:\,a\leq x\leq b\}$. We can calculate
$\mathbb{P}(\mathsf{rand}(D)\in[a,b])=\mathsf{cdf}(D,b)-\mathsf{cdf}(D,a)$
using the cumulative density function $\mathsf{cdf}(D,x)$. This probability
allows us to update the \code{weight} of the trial. For instance,
a call \code{observe(Normal(2.0,0.1), {[}a,b{]})} can be executed
as \code{weight {*}= normalcdf(2.0,0.1,b) - normalcdf(2.0,0.1,a)}
where $normalcdf(\mu,\sigma,x)$ is the cumulative density function
for the normal distribution.

Notice how this change from probability densities to probabilities
prevents unit anomalies: if we change the variables $a,b$ from meters
to centimeters, then we must write \code{observe(Normal(200,10), {[}a,b{]})},
which gets executed as \code{weight {*}= normalcdf(200,10,b) - normalcdf(200,10,a)}.
We introduced a factor $100$ to convert $\mu$ and $\sigma$ from
meters to centimeters. This conversion ensures that the result of
the program remains unchanged, because $\mathsf{normalcdf}(r\mu,r\sigma,rx)=\mathsf{normalcdf}(\mu,\sigma,x)$
for all $r>0$. Hence the computed \code{weight} will be exactly
the same whether we work with meters or centimeters. On the other
hand, for the probability density function it is \textbf{not} the
case that $\mathsf{normalpdf}(r\mu,r\sigma,rx)=\mathsf{normalpdf}(\mu,\sigma,x)$.
It is precisely this lack of invariance that causes unit anomalies
with probability densities.

\subsection{Conditioning on Measure Zero Events as a Limit of Positive Measure
Events}

We can approximate the old \code{observe(D,x)} behavior with \code{observe(D,I)}
by choosing $I=[x-\frac{1}{2}w,x+\frac{1}{2}w]$ to be a very small
interval of width \code{w} around \code{x} (taking \code{w} to
be a small number, such as \code{w = 0.0001}). This has two important
advantages over \code{observe(D,x)}:
\begin{enumerate}
\item We no longer get unit anomalies or other paradoxes; if we change the
units of \code{x}, we must also change the units of \code{w}, which
keeps the \code{weight} the same.
\item Unlike for \code{observe(D,x)}, we have an unambiguous probabilistic
and rejection sampling interpretation of \code{observe(D,I)} for
intervals of nonzero width, because the event being conditioned on
has nonzero measure.
\end{enumerate}
However, the number \code{w = 0.0001} is rather arbitrary. We would
like to let $w\to0$ and recover the functionality of \code{observe(D,x)}
to condition on an exact value. With sufficiently small \code{w}
we can get arbitrarily close, but we can never recover its behavior
exactly.

We therefore parameterize probabilistic programs by a dimensionless
parameter \code{eps}. The BMI example then becomes:
\begin{verbatim}

  function bmi_example(eps)
     h = rand(Normal(170, 50))
     w = rand(Normal(70, 30))
     if rand(Bernoulli(0.5))
       observe(Normal(200, 10), (h, A*eps))
     else
        observe(Normal(90, 5), (w, B*eps))
     end
     return w / h^2
  end

\end{verbatim}
Since \code{eps} is dimensionless, we can not simply use \code{eps}
as the width of the intervals: because \code{h} is in $cm$, the
width of the interval around \code{h} has to be in $cm$, and the
width of the interval around \code{w} has to be in $kg$. We are
forced to introduce a constant \code{A} with units $cm$ and a constant
\code{B} with units $kg$ that multiply \code{eps} in the widths
of the intervals in the observes.

We could now run importance sampling on \code{bmi\_example(eps)}
for \code{n=10000} trials for \code{eps=0.1}, \code{eps=0.01},
\code{eps=0.001} and so on, to see what value it converges to. If
we run each of these independently, then the \code{rand} calls will
give different results, so there will be different randomness in each
of these, and it may be difficult to see the convergence. In order
to address this, we can run the program with different values of \code{eps}
but with the same random seed for the random number generator. This
will make the outcomes of the rand calls the same regardless of the
value of \code{eps}. In fact, for a given random seed, the result
of running importance sampling for a given number of trials will be
a deterministic function \code{f(seed,eps)} of the random \code{seed}
and \code{eps}

If we assume that the program uses $\mathsf{eps}=\epsilon$ only in
the widths of the intervals, and not in the rest of the program, then
for a fixed \code{seed}, the function $f(\mathsf{seed},\epsilon)$
will be a function of $\epsilon$ of a specific form, because importance
sampling compute
\begin{align*}
f(\mathsf{seed},\epsilon) & =\frac{\sum_{k=1}^{N}(\mathit{weight}_{k}(\epsilon))\cdot(value_{k})}{\sum_{k=1}^{N}(\mathit{weight}_{k}(\epsilon))}
\end{align*}
In this fraction, the $weight_{k}$ are a function of $\epsilon$,
but the $value_{k}$ are independent of $\epsilon$ if $\epsilon$
only occurs inside the widths of intervals. Since the weight gets
multiplied by $P(D,I)$ on each \code{observe(D,I)}, the $weight_{k}(\epsilon)$
is of a very specific form:
\begin{align*}
weight_{k} & (\epsilon)=C\cdot P(D_{1},(x_{1},w_{1}\epsilon))\cdots P(D_{n},(x_{n},w_{n}\epsilon))
\end{align*}
where the constant $C$ contains all the probabilities accumulated
from observes that did not involve $\epsilon$, multiplied by a product
of probabilities that did involve $\epsilon$. Since $P(D,(x,w\epsilon))=\mathsf{cdf}(D,x+\frac{1}{2}w\epsilon)+\mathsf{cdf}(D,x-\frac{1}{2}w\epsilon)$,
we could, in principle determine the precise function $weight_{k}(\epsilon)$
and hence $f(\mathsf{seed},\epsilon)$ for any given seed. We could
then, in principle, compute the exact limit of this function as $\epsilon\to0$,
with a computer algebra system. This is, of course, impractical. The
next section shows that we can compute the limit efficiently by doing
arithmetic with infinitesimal numbers.

\section{Using Infinitesimal Numbers to Handle Measure-Zero Observations\label{sec:USING-INFINITESIMAL-NUMBERS}}

In order to recover the behavior of the old \code{observe(D,x)} using
\code{observe(D,I)} with an interval $I=[x-\frac{1}{2}w,x+\frac{1}{2}w]$,
we want to take the limit $w\to0$, to make $[x-\frac{1}{2}w,x+\frac{1}{2}w]$
an infinitesimally small interval around $x$. We accomplish this
using symbolic infinitesimal numbers\footnote{In the philosophy literature there has been work on using non-standard
analysis and other number systems to handle probability 0 events,
see \citep{pedersen}and \citep{however} and references therein. } of the form $r\epsilon^{n}$, where $r\in\mathbb{R}$ and $n\in\mathbb{Z}$.
We allow $n<0$, so that $r\epsilon^{n}$ can also represent ``infinitely
large'' numbers as well as ``infinitesimally small'' numbers. We
will not make use of this possibility, but it makes the definitions
and proofs more general and more uniform.\footnote{These infinitesimal numbers may be viewed as the leading terms of
Laurent series. This bears some resemblance to the dual numbers used
in automatic differentiation, which represent the constant and linear
term of the Taylor series. In our case we only have the first nonzero
term of the Laurent series, but the order of the term is allowed to
vary.}
\begin{definition}
An infinitesimal number is a pair $(r,n)\in\mathbb{R}\times\mathbb{Z}$,
which we write as $r\epsilon^{n}$.\footnote{The exponent $n$ of $\epsilon$ will play the same role as the number
of densities $d$ in lexicographic likelihood weighting\citep{pmlr-v80-wu18f}.}
\end{definition}
\noindent The infinitesimals of the form $r\epsilon^{0}$ correspond
to the real numbers.
\begin{definition}
Addition, subtraction, multiplication, and division on those infinitesimal
numbers are defined as follows:\label{def:infarith}
\begin{align*}
r\epsilon^{n}\pm s\epsilon^{k} & =\begin{cases}
(r\pm s)\epsilon^{n} & \text{if }n=k\\
r\epsilon^{n} & \text{if }n<k\\
\pm s\epsilon^{k} & \text{if }n>k
\end{cases}\\
(r\epsilon^{n})\cdot(s\epsilon^{k}) & =(r\cdot s)\epsilon^{n+k}\\
(r\epsilon^{n})/(s\epsilon^{k}) & =\begin{cases}
(r/s)\epsilon^{n-k} & \text{if \ensuremath{s\neq0}}\\
\text{undefined} & \text{if \ensuremath{s=0}}
\end{cases}
\end{align*}
Like ordinary division, division of infinitesimals is a partial function,
which is \code{undefined} if the denominator is \emph{exactly} zero.
\end{definition}
These rules may be intuitively understood by thinking of $\epsilon$
as a very small number; e.g. if $n<k$ then $\epsilon^{k}$ will be
negligible compared to $\epsilon^{n}$, which is why we define $r\epsilon^{n}+s\epsilon^{k}=r\epsilon^{n}$
in that case, and keep only the lowest order term.

We represent intervals $[x-\frac{1}{2}w,x+\frac{1}{2}w]$ as midpoint-width
pairs $(x,w)$, where $w$ may be an infinitesimal number.
\begin{definition}
\label{Pdef}If $D$ is a continuous distribution, we compute the
probability $P(D,(x,w))$ that $X\sim D$ lies in the interval $(x,w)$
as:
\begin{align}
P(D,(x,w)) & =\begin{cases}
\mathsf{cdf}(D,x+\frac{1}{2}r)-\mathsf{cdf}(D,x-\frac{1}{2}r) & \text{if }w=r\epsilon^{0}\text{ is not infinitesimal}\\
\mathsf{pdf}(D,x)\cdot r\epsilon^{n} & \text{if }w=r\epsilon^{n}\text{ is infinitesimal }(n>0)
\end{cases}\label{eq:P}
\end{align}
Where $\mathsf{cdf}(D,x)$ and $\mathsf{pdf}(D,x)$ are the cumulative
and probability density functions, respectively.
\end{definition}
Note that the two cases agree in the sense that if $w$ is very small,
then
\begin{align*}
\mathsf{cdf}(D,x+\frac{1}{2}w)-\mathsf{cdf}(D,x-\frac{1}{2}w) & \approx\frac{d}{dx}\mathsf{cdf}(D,x)\cdot w=\mathsf{pdf}(D,x)\cdot w
\end{align*}

\begin{definition}
We say that $f(x)$ is a ``probability expression'' in the variable
$x$ if $f(x)$ is defined using the operations $+,-,\cdot,/$, constants,
and $P(D,(s,rx))$ where $r,s\in\mathbb{R}$ are constants, and $D$
is a probability distribution with differentiable cdf.\label{def:probexpr}

We can view $f$ as a function from reals to reals (on the domain
on which it is defined, that is, excluding points where division by
zero happens), or as a function from infinitesimals to infinitesimals
by re-interpreting the operations in infinitesimal arithmetic. The
value of $f(\epsilon)$ on the symbolic infinitesimal $\epsilon$
tells us something about the limiting behavior of $f(x)$ near zero:
\end{definition}
\begin{theorem}
If $f(x)$ is a probability expression, and if evaluation of $f(\epsilon)$
is not \code{undefined}, and $f(\epsilon)=r\epsilon^{n}$, then $\lim_{x\to0}\frac{f(x)}{x^{n}}=r$.\label{thm:lim}
\end{theorem}
Note that the theorem only tells us that $\lim_{x\to0}\frac{f(x)}{x^{n}}=r$
\emph{if} $f(\epsilon)$ evaluates to $r\epsilon^{n}$ with infinitesimal
arithmetic. If evaluating $f(\epsilon)$ results in division by zero,
then the theorem does not give any information. In fact, the converse
of the theorem does \emph{not} hold: it may be that $\lim_{x\to0}\frac{f(x)}{x^{n}}=r$
but evaluating $f(\epsilon)$ results in division by zero.
\begin{proof}
By induction on the structure of the expression.\\
We know that evaluation of $f(\epsilon)$ did not result in division
by zero, and $f(\epsilon)=r\epsilon^{n}$. We need to show that $\lim_{x\to0}\frac{f(x)}{x^{n}}=r$.
\begin{itemize}
\item If $f(x)$ is a constant $r$, then we have $f(\epsilon)=r\epsilon^{0}$,
and indeed $\lim_{x\to0}\frac{f(x)}{x^{0}}=\lim_{x\to0}f(x)=r$.
\item If $f(x)=P(D,(s,rx))$. Now $f(\epsilon)=\mathsf{pdf}(D,s)\cdot r\epsilon$,
and
\begin{align*}
\mathsf{pdf}(D,s)\cdot r & =r\frac{d}{dx}[\mathsf{cdf}(D,x)]_{x=s}\\
 & =r\lim_{x\to0}\frac{\mathsf{cdf}(D,s+x)-\mathsf{cdf}(D,s-x)}{2x}\\
 & =\lim_{x'\to0}\frac{\mathsf{cdf}(D,s+\frac{1}{2}rx')-\mathsf{cdf}(D,s-\frac{1}{2}rx')}{x'}\\
 & =\lim_{x'\to0}\frac{P(D,(s,rx'))}{x'}
\end{align*}
\item If $f(x)=g(x)+h(x)$. Since evaluation of $f(\epsilon)$ did not result
in division by zero, neither did evaluation of the subexpressions
$g(\epsilon)$ and $h(\epsilon)$, so $g(\epsilon)=r_{1}\epsilon^{n_{1}}$
and $h(\epsilon)=r_{2}\epsilon^{n_{2}}$ for some $r_{1},r_{2},n_{1},n_{2}$.
Therefore, by the induction hypothesis we have $\lim_{x\to0}\frac{g(x)}{x^{n_{1}}}=r_{1}$
and $\lim_{x\to0}\frac{h(x)}{x^{n_{2}}}=r_{2}$.
\item Case $n_{1}=n_{2}$: Now $f(\epsilon)=(r_{1}+r_{2})\epsilon^{n_{1}}$,
and we have
\begin{align*}
\lim_{x\to0}\frac{f(x)}{x^{n_{1}}}= & \lim_{x\to0}\frac{g(x)+h(x)}{x^{n_{1}}}=\lim_{x\to0}\frac{g(x)}{x^{n_{1}}}+\lim_{x\to0}\frac{h(x)}{x^{n_{1}}}=r_{1}+r_{2}
\end{align*}
\item Case $n_{1}<n_{2}$: Now $f(\epsilon)=r_{1}\epsilon^{n_{1}}$, and
since $\lim_{x\to0}\frac{h(e)}{x^{n_{2}}}=r_{2}$ we have
\begin{align*}
0 & =0\cdot r_{2}=(\lim_{x\to0}x^{n_{2}-n_{1}})\cdot(\lim_{x\to0}\frac{h(x)}{x^{n_{2}}})=\text{\ensuremath{\lim_{x\to0}\frac{x^{n_{2}-n_{1}}h(x)}{x^{n_{2}}}}}=\lim_{x\to0}\frac{h(x)}{x^{n_{1}}}
\end{align*}
Therefore
\begin{align*}
\lim_{x\to0}\frac{f(x)}{x^{n_{1}}}= & \lim_{x\to0}\frac{g(x)+h(x)}{x^{n_{1}}}=\lim_{x\to0}\frac{g(x)}{x^{n_{1}}}+\lim_{x\to0}\frac{h(x)}{x^{n_{1}}}=r_{1}
\end{align*}
\item Case $n_{1}>n_{2}$. Analogous to the previous case.
\item If $f(x)=g(x)-h(x)$. Analogous to the case for addition.
\item If $f(x)=g(x)\cdot h(x)$. Since evaluation of $f(\epsilon)$ did
not result in division by zero, neither did evaluation of the subexpressions
$g(\epsilon)$ and $h(\epsilon)$, so $g(\epsilon)=r_{1}\epsilon^{n_{1}}$
and $h(\epsilon)=r_{2}\epsilon^{n_{2}}$ for some $r_{1},r_{2},n_{1},n_{2}$.
Therefore, by the induction hypothesis we have $\lim_{x\to0}\frac{g(x)}{x^{n_{1}}}=r_{1}$
and $\lim_{x\to0}\frac{h(x)}{x^{n_{2}}}=r_{2}$. Then
\begin{align*}
\lim_{x\to0}\frac{f(x)}{x^{n_{1}+n_{2}}} & =\lim_{x\to0}\frac{g(x)}{x^{n_{1}}}\cdot\frac{h(x)}{x^{n_{2}}}=(\lim_{x\to0}\frac{g(x)}{x^{n_{1}}})\cdot(\lim_{x\to0}\frac{h(x)}{x^{n_{2}}})=r_{1}\cdot r_{2}
\end{align*}
\item If $f(x)=g(x)/h(x)$. Since evaluation of $f(\epsilon)$ did not result
in division by zero, neither did evaluation of the subexpressions
$g(\epsilon)$ and $h(\epsilon)$, so $g(\epsilon)=r_{1}\epsilon^{n_{1}}$
and $h(\epsilon)=r_{2}\epsilon^{n_{2}}$ for some $r_{1},r_{2},n_{1},n_{2}$.
Therefore, by the induction hypothesis we have $\lim_{x\to0}\frac{g(x)}{x^{n_{1}}}=r_{1}$
and $\lim_{x\to0}\frac{h(x)}{x^{n_{2}}}=r_{2}$. By the assumption
that no division by exactly zero occurred in the evaluation of $f(\epsilon)$,
we have $r_{2}\neq0$. Then
\begin{align*}
\lim_{x\to0}\frac{f(x)}{x^{n_{1}+n_{2}}} & =\lim_{x\to0}\frac{g(x)}{x^{n_{1}}}/\frac{h(x)}{x^{n_{2}}}=(\lim_{x\to0}\frac{g(x)}{x^{n_{1}}})/(\lim_{x\to0}\frac{h(x)}{x^{n_{2}}})=r_{1}/r_{2}
\end{align*}
\end{itemize}
This finishes the proof.
\end{proof}

\subsubsection*{Some subtleties of limits and infinitesimals}

In order to think about infinitesimals one must first choose a function
$f(x)$ of which one wishes to learn something about the limit as
$x\to0$. Thinking about infinitesimal arithmetic independent of such
a function leads to confusion. Furthermore, the result of evaluating
$f(\epsilon)$ depends not just on $f(x)$ as a function on real numbers,
but also on the arithmetic expression used for computing $f$. Consider
the functions $f,g$:
\begin{align*}
f(x) & =5\cdot x^{2}+0\cdot x\\
g(x) & =5\cdot x^{2}
\end{align*}
As functions on real numbers, $f=g$, but nevertheless, with infinitesimal
arithmetic their results differ:
\begin{align*}
f(\epsilon) & =0\cdot\epsilon^{1}\\
g(\epsilon) & =5\cdot\epsilon^{2}
\end{align*}
Applying the theorem to these results gives the following limits for
$f$ and $g$:
\begin{align*}
\lim_{x\to0}\frac{f(x)}{x} & =0\\
\lim_{x\to0}\frac{g(x)}{x^{2}} & =5
\end{align*}
Both of these limits are correct, but this example shows that \emph{which}
limit the theorem says something about may depend on how the function
is computed. The limit for $g$ gives more information than the limit
for $f$; the limit for $f$ is conservative and doesn't tell us as
much as the limit for $g$ does. Fortunately, this won't be a problem
for our use case: we intend to apply the theorem to the weighted average
of importance sampling, where the probabilities may be infinitesimal
numbers. In this case the power of $\epsilon$ of the numerator and
denominator are always the same, so the final result will always have
power $\epsilon^{0}$, and the theorem will then tell us about the
limit $\lim_{x\to0}\frac{f(x)}{x^{0}}=\lim_{x\to0}f(x)$.

Another subtlety is that the converse of the theorem does not hold.
It is possible that $\lim_{x\to0}\frac{f(x)}{x^{n}}=r$, but evaluation
of $f(\epsilon)$ with infinitesimal arithmetic results in division
by exactly zero. An example is $f(x)=\frac{x^{2}}{(x+x^{2})-x}$.
We have $\lim_{x\to0}f(x)=1$, but when evaluating $f(\epsilon)=\frac{\epsilon^{2}}{(\epsilon+\epsilon^{2})-\epsilon}$,
division by zero occurs, because we have the evaluation sequence:
\begin{align*}
\frac{\epsilon^{2}}{(\epsilon+\epsilon^{2})-\epsilon}\to\frac{\epsilon^{2}}{\epsilon-\epsilon}\to\frac{\epsilon^{2}}{0}\to\text{undefined}
\end{align*}
If we used full Laurent series $a_{k}\epsilon^{k}+a_{k+1}\epsilon^{k+1}+\dots$
as our representation for infinitesimal numbers, then we would potentially
be able to compute more limits, even some of those where exact cancellation
happens in a denominator. Keeping only the first term is sufficient
for our purposes, and more efficient, because our infinitesimal numbers
are pairs $(r,n)$ of a real (or floating point) number $r$ and an
integer $n$, whereas Laurent series are infinite sequences of real
numbers $(a_{k},a_{k+1},\dots)$.

The lemmas about computing limits have the form ``For all $a,b\in\mathbb{R}$,
if $\lim_{x\to0}f(x)=a$, and $\lim_{x\to0}g(x)=b$, and $b\neq0$,
then $\lim_{x\to0}\frac{f(x)}{g(x)}=\frac{\lim_{x\to0}f(x)}{\lim_{x\to0}g(x)}$''.
It is \emph{not} true in general that $\lim_{x\to0}\frac{f(x)}{g(x)}=\frac{\lim_{x\to0}f(x)}{\lim_{x\to0}g(x)}$.
It is possible that the limit on the left hand side exists, even when
the limits on the right hand side fail to exist, or when the right
hand side is $\frac{0}{0}$. Therefore, in order to apply these theorems
about limits, we must know that the right hand side is not undefined,
prior to applying such a lemma. In the proof above, the existence
of the limits follows from the induction hypothesis, and that the
denominator is nonzero follows from the assumption that division by
zero does not occur. This is why we must assume that no division by
exactly zero occurs in the evaluation of $f(\epsilon)$ with infinitesimal
arithmetic, and it is also why the converse of the theorem does not
hold.

\subsection{Intervals of Infinitesimal Width Make Paradoxes Disappear}

The proposed observe construct allows finite width intervals \code{observe(D,(a,w))}
where \code{w} is an expression that returns a number, as well as
infinitesimal width intervals, as in \code{observe(D,(a,w{*}eps))}
where \code{w} is some expression that returns a number and \code{eps}
is the symbolic infinitesimal $\epsilon$. It is possible to allow
higher powers of \code{eps} to occur directly in the source program,
and it is possible to allow \code{eps} to occur in other places than
in widths of intervals, but for conceptual simplicity we shall assume
it doesn't, and that \code{observe} is always of one of those two
forms. That is, we will assume that \code{eps} is only used in order
to translate exact conditioning \code{observe(D,x)} to \code{observe(D,(x,w{*}eps))}.

We translate the example from the introduction as follows:
\begin{verbatim}

  h = rand(Normal(170, 50))
  if rand(Bernoulli(0.5))
     observe(Normal(200, 10), (h,w*eps))
  end

\end{verbatim}
Where the pair \code{(h,w{*}eps)} represents an interval of width
\code{w{*}eps} centered around \code{h}, in order to condition on
the observation to be ``exactly $h$''.

Let us now investigate the meaning of this program according to the
rejection sampling interpretation of \code{observe}. Assuming the
coin flip results in $true$, we reject the trial if the sample from
$Normal(200,10)$ does not fall in the interval $[h-\frac{1}{2}w\epsilon,h+\frac{1}{2}w\epsilon]$.
If the coin flip results in $false$, we always accept the trial.
If we let $\epsilon\to0$ then the probability of rejecting the trial
goes to $1$ if the coin flips to $true$, so almost all successful
trials will be those where the coin flipped to $false$. Therefore
the expected value of $h$ converges to $170$ as $\epsilon\to0$,
and expected value of running this program should be $170$.\\

We translate the BMI example as follows:
\begin{verbatim}

  h = rand(Normal(170, 50))
  w = rand(Normal(70, 30))
  if rand(Bernoulli(0.5))
     observe(Normal(200, 10), (h, A*eps))
  else
     observe(Normal(90, 5), (w, B*eps))
  end
  bmi = w / h^2

\end{verbatim}
Where \code{A} and \code{B} are constants with units $cm$ and $kg$,
respectively. The units force us to introduce these constants: since
\code{(h, A{*}eps)} represents an interval centered at \code{h}
(in cm), the width \code{A{*}eps} must also be a quantity in $cm$.
If we change the units of \code{h} or \code{w}, we also need to
change the units of \code{A} or \code{B}. If we change the units
of \code{h} and \code{A} from centimeters to meters, the numerical
value of \code{h} and \code{A} will both get multiplied by $\frac{1}{100}$.
This additional factor for \code{A{*}eps}, which cannot be provided
in the original non-interval type of \code{observe(D,x)} statement,
is what will make this program behave consistently under change of
units.

Both branches of the if statement contain observes with intervals
of infinitesimal width, so with rejection sampling both branches will
be rejected with probability 1, regardless of the outcome of the coin
flip. We must therefore interpret the example with \code{eps} tending
to $0$, but not being exactly $0$. If we chose \code{A} to be 1
meter, and \code{B} to be 1 kg, and change \code{B} to be 1000 kg,
then the observe in the else branch is 1000x more likely to succeed
compared to before, because the width of the interval goes from \code{1{*}eps}
to \code{1000{*}eps}. If we made this change then most of the successful
trials would be trials where the coin flipped to \code{false}. Thus
even in the infinitesimal case, the \emph{relative} sizes of the intervals
matter a great deal. The relative sizes of the intervals are an essential
part of the probabilistic program, and omitting them will inevitably
lead to unit anomalies, because changing units also requires resizing
the intervals by a corresponding amount (by 1000$\times$ in case
we change $w$ from $kg$ to $g$). If we do not resize the intervals,
that changes the relative rejection rates of the branches, or the
relative weights of the trials, and thus the estimated expectation
value $\mathbb{E}[bmi]$. As Jaynes notes, conditioning on measure-zero
events is ambiguous; even though in the limit the intervals \code{(w,1{*}eps)}
and \code{(w,1000{*}eps)} both tend to the singleton set \code{\{w\}},
relative to the interval \code{(h,A{*}eps)} it matters \emph{which}
of these limits is intended, and the final result will and must depend
on which limit was specified.

We translate the third example as follows:
\begin{verbatim}

  x = rand(Normal(10,5))
  observe(Normal(15,5), (x,eps))
  return exp(x)

\end{verbatim}
After a parameter transformation from $x$ to $\exp(x)$ we get the
following program:
\begin{verbatim}

  exp_x = rand(LogNormal(10,5))
  observe(LogNormal(15,5), (exp_x,exp_x*eps))
  return exp_x

\end{verbatim}
Note that the width of the interval is now \code{exp\_x{*}eps} and
not simply \code{eps}. In general, if we apply a differentiable function
$f$ to an interval of width $\epsilon$ around $x$, we obtain an
interval of width $f'(x)\epsilon$ around $f(x)$. If we take the
exponential of an interval of small width $\epsilon$ around $x$,
we get an interval of width $\exp(x)\epsilon$ around $\exp(x)$,
\textbf{not }an interval of width $\epsilon$ around $\exp(x)$. Both
of these programs should give the same estimate for the expectation
value of $\exp(x)$, so that infinitesimal width intervals allow us
to correctly express non-linear parameter transformations without
running into Borel-Komolgorov-type paradoxes.

\subsection{On the Statistical Meaning of Conditioning With Intervals and ``Soft
Conditioning''}

It is debatable whether conditioning on small but finite width intervals
is preferable to conditioning on measure zero events. Real measurement
devices do not measure values to infinite precision. If a measurement
device displays 45.88, we might take that to mean an observation in
the interval $[45.875,45.885]$. The measurement may in addition measure
the true value \code{x} plus some \code{Normal(0,sigma)} distributed
noise rather than the true value \code{x}. In this case it might
be appropriate to use \code{observe(Normal(x,sigma), (45.88, 0.01))}.
The finite precision of the device and its noisy measurement are in
principle two independent causes of uncertainty. The rejection sampling
interpretation of this program is that we first sample a value from
\code{Normal(x,sigma)} and then continue with the current trial if
this lies in the interval $[45.875,45.885]$, which matches the two
sources of uncertainty. An argument for using infinitesimal width
intervals is that \code{observe} on a finite interval requires the
evaluation of the distribution's CDF, which is usually more complicated
and expensive to compute than the distribution's PDF.

The term ``soft conditioning'' is sometimes used for \code{observe(D,x)}
statements, particularly when the distribution \code{D} is the normal
distribution. This term can be interpreted as an alternative to the
rejection sampling interpretation in several ways:
\begin{enumerate}
\item Rather than conditioning on \code{x} being exactly \code{y}, we
instead condition on \code{x} being ``roughly'' \code{y}.
\item The statement \code{observe(D,x)} means that we continue with the
current trial with probability \code{pdf(D,x)} and reject it otherwise.
\end{enumerate}
We argue that neither of these interpretations is completely satisfactory.
For (1) it is unclear what the precise probabilistic meaning of conditioning
on $x$ being ``roughly'' $y$ is. One possible precise meaning
of that statement is that we reject the trial if the difference $|x-y|$
is too large, and continue otherwise, but this is not what a statement
such as \code{observe(Normal(y,0.01), x)} does. Rather, it weighs
trials where $x$ is close to $y$ higher, and smoothly decreases
the weight as the distance between $x$ and $y$ gets larger. It may
seem that (2) makes this idea precise, but unfortunately \code{pdf(D,x)}
is not a probability but a probability density, and can even have
units or be larger than $1$. Furthermore, the statement ``continue
with the current trial with probability \code{pdf(D,x)}'' seems
to have nothing to do with the distribution \code{D} as a probability
distribution, and instead seems to be a statement that suggests that
the statistical model is a biased coin flip rather than drawing a
sample from \code{D}. Indeed, under our rejection sampling interpretation,
if one wants to have a program whose statistical model is about coin
flips, one can use the program \code{observe(Bernoulli(f(x)), true)}.
That program \emph{does} mean ``flip a biased coin with heads probability
\code{f(x)} and continue with the current trial if the coin landed
heads''. This makes sense for any function \code{f(x)} provided
the function gives us a probability in the range $[0,1]$. If that
function has a roughly bump-like shape around \code{y}, then this
will indeed in some sense condition on \code{x} being roughly \code{y}.
The function $C\exp((x-A)^{2}/B)$ similar to the PDF of the normal
distribution does have a bump-like shape around $A$, so it is possible
to use that function for \code{f}, if one makes sure that $B$ and
$C$ are such that it is unitless and everywhere less than 1 (note
that this normalization is not the same as the normalization that
makes its integral sum to $1$).

We therefore suggest to stick with the rejection sampling interpretation
of observe statements, and suggest that a statistician who wants to
do ``soft conditioning'' in the senses (1) and (2) writes their
probabilistic program using \code{observe(Bernoulli(f(x)), true)}
where \code{f} is a function of the desired soft shape rather than
\code{observe(D,x)} where the PDF of \code{D} has that shape.

\subsection{Importance Sampling with Infinitesimal Probabilities\label{subsec:Importance-sampling-with}}

To do importance sampling for programs with infinitesimal width intervals
we need to change almost nothing. We execute a call \code{observe(D,I)}
as \code{weight {*}= P(D,I)} where \code{P(D,I)} has been defined
in (\ref{eq:P}). Since \code{P(D,I)} returns an infinitesimal number
if the width of \code{I} is infinitesimal, the computed \code{weight}
variable will now contain a symbolic infinitesimal number $r\epsilon^{n}$
(where $n$ is allowed to be $0$), rather than a real number. It
will accumulate the product of some number of ordinary probabilities
(for \code{observe} on discrete distributions or continuous distributions
with an interval of finite width) and a number of infinitesimal probabilities
(for \code{observe} on continuous distributions with intervals of
infinitesimal width).

We now simply evaluate the estimate for $\mathbb{E}[V]$ using the
usual weighted average formula, with infinitesimal arithmetic
\begin{align}
\mathbb{E}[V] & \approx\frac{\sum_{k=0}^{N}(weight_{k})\cdot(V_{k})}{\sum_{k=0}^{N}(weight_{k})}\label{eq:est}
\end{align}
In the denominator we are adding numbers of the form \emph{$weight_{k}=w_{k}\epsilon^{n_{k}}$.}
Only the numbers with the minimum value $n_{k}=n_{min}$ matter; the
others are infinitesimally small compared to those, and do not get
taken into account due to the definition of $(+)$ on infinitesimal
numbers. The same holds for the numerator: the values $V_{k}$ associated
with weights that are infinitesimally smaller do not get taken into
account (an optimized implementation could reject a trial as soon
as \code{weight} becomes infinitesimally smaller than the current
sum of accumulated weights, since those trials will never contribute
to the estimate of $\mathbb{E}[V]$). Therefore the form of the fraction
is
\begin{align*}
\mathbb{E}[V]\approx\frac{A\epsilon^{n_{min}}}{B\epsilon^{n_{min}}} & =\frac{A}{B}\epsilon^{n_{min}-n_{min}}=\frac{A}{B}\epsilon^{0}
\end{align*}
that is, the infinitesimal factors cancel out in the estimate for
$\mathbb{E}[V]$, and we obtain a non-infinitesimal result.

We shall now suppose that the symbolic infinitesimal \code{eps} only
occurs in the width of intervals in\code{observe(D,(x,r{*}eps))}
calls, and not, for instance, in the return value of the probabilistic
program. In this case, the estimate (\ref{eq:est}) of $\mathbb{E}[V]$
satisfies the conditions of Theorem \ref{thm:lim}. The calculated
estimate may be viewed as a probability expression $f(\epsilon)$
of $\epsilon$ (Definition \ref{def:probexpr}), and since $f(\epsilon)=\frac{A}{B}\epsilon^{0}$,
the theorem implies that $\lim_{x\to0}f(x)=\frac{A}{B}$. Therefore
the estimate calculated by importance sampling with infinitesimal
arithmetic indeed agrees with taking the limit $\epsilon\to0$. Figure
\ref{graphs} shows three example probabilistic programs that are
parameterized by the interval width. The blue lines show several runs
of the probabilistic program as a function of the interval width,
and the orange line shows the result when taking the width to be $\epsilon$.
Taking the width to be exactly 0 results in division by zero in the
weighted average, but taking it to be $\epsilon$ correctly computes
the limit: the blue lines converge to the orange lines as the width
goes to 0.

\begin{figure}
\label{graphs}

\begin{minipage}[c][1\totalheight][t]{0.42\textwidth}%
\begin{verbatim}
  function example1(width)
    h = rand(Normal(1.70, 0.2))
    w = rand(Normal(70, 30))
    if rand(Bernoulli(0.5))
      observe(Normal(2.0,0.1),
        Interval(h,10*width))
    else
      observe(Normal(90,5),
        Interval(w,width))
    end
    return w / h^2
  end
\end{verbatim}
\end{minipage}%
\begin{minipage}[c][1\totalheight][t]{0.58\textwidth}%
\includegraphics[width=1\textwidth]{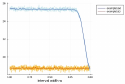}%
\end{minipage}

\begin{minipage}[c][1\totalheight][t]{0.42\textwidth}%
\begin{verbatim}
  function example2(width)
    h = rand(Normal(1.7,0.5))
    if rand(Bernoulli(0.5))
      observe(Normal(2.0,0.1),
        Interval(h,width))
    end
    return h
  end
\end{verbatim}
\end{minipage}%
\begin{minipage}[c][1\totalheight][t]{0.58\textwidth}%
\includegraphics[width=1\textwidth]{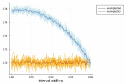}%
\end{minipage}

\begin{minipage}[c][1\totalheight][t]{0.42\textwidth}%
\begin{verbatim}
  function example3(width)
    x = rand(Normal(10,5))
    observe(Normal(15,5),
      Interval(x,width))
    return x
  end
\end{verbatim}
\end{minipage}%
\begin{minipage}[c][1\totalheight][t]{0.58\textwidth}%
\includegraphics[width=1\textwidth]{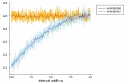}%
\end{minipage}

\caption{Three example programs evaluated with finite width intervals with
width going to zero (blue curves) and with infinitesimal width (orange
curves). The finite width result correctly converges to the infinitesimal
result in the limit $w\to0$.}

\vfill{}
\end{figure}

\subsection{The Correspondence Between Observe on Points and Observe on Intervals}

We may take a program written using \code{observe(D,x)} with exact
conditioning on points, and convert it to our language by replacing
such calls with \code{observe(D,(x,w{*}eps))} where \code{w} is
some constant to make the units correct. For programs that exhibit
a paradox of type 1 by executing a different number of observes depending
on the outcome of calls to \code{rand}, the computed expectation
values will change. However, for programs that always execute the
same number of \code{observe} calls, regardless of the outcome of
\code{rand} calls, the computed expectation values will not be affected
by this translation. To see this, note that a call to \code{observe(D,x)}
will multiply \code{weight {*}= pdf(D,x)}, whereas \code{observe(D,(x,w{*}eps))}
will multiply \code{weight {*}= pdf(D,x){*}w{*}eps}. Thus if the
observe calls are the same in all trials, the only difference is that
weight will contain an extra factor of $w\epsilon$ in all trials.
The net result is that both the numerator and denominator in the weighted
average get multiplied by the factor $w\epsilon$, which has no effect.
Thus this translation is conservative with respect to the old semantics,
in the sense that it does not change the result that already well-behaved
probabilistic programs compute.

\subsection{Parameter Transformations as a Language Feature}

The three paradoxes we identified all have to do with parameter transformations.
We explicitly add parameter transformations as a language feature.
A parameter transformation $T$ allows us to transform a probability
distribution $D$ to $T(D)$, such that sampling from $T(D)$ is the
same as sampling from $D$ and then applying the function $T$ to
the result. In order to ensure that the distribution $T(D)$ has a
probability density function we require $T$ to be continuously differentiable.
We can also use a parameter transformation to transform an interval
from $I$ to $T(I)=\{T(x)\,:\,x\in I\}$ which contains all the numbers
$T(x)$ for $x\in I$. In order to ensure that the transformed interval
is again an interval, we require that $T$ is monotone, that is, whenever
$a<b$ we also have $T(a)<T(b)$. In this case, $T$'s action on an
interval $[a,b]$ is simple: $T([a,b])=[T(a),T(a)]$.
\begin{definition}
A parameter transformation $T:\mathbb{R_{A}}\to\mathbb{R_{B}}$ is
a continuously differentiable function with $T'(x)>0$ for all $x\in\mathbb{R_{A}}$,
where $\mathbb{R_{A}\subseteq R}$ and $\mathbb{R_{B}\subseteq R}$
are intervals representing its domain and range.
\end{definition}
A strictly monotone function has an inverse on its range, so parameter
transformations have an inverse $T^{-1}$ and $T^{-1}(y)=T'(T^{-1}(y))^{-1}>0$,
so the inverse of a parameter transformation is again a parameter
transformation.
\begin{example}
The function $T_{1}(x)=\exp(x)$ is a parameter transformation $T_{1}:(-\infty,\infty)\to[0,\infty)$.
The function $T_{2}(x)=100x$ is a parameter transformation $T_{2}:(-\infty,\infty)\to(-\infty,\infty)$.
\end{example}
The transformation $T_{1}$ can be used to convert decibels to energy
density, and $T_{2}$ can be used to convert meters to centimeters.

Probability distributions need to support 3 operations: random sampling
with \code{rand(D)}, computing the CDF with \code{cdf(D,x)} and
computing the PDF with \code{pdf(D,x)}. We define these operations
for the transformed distribution $T(D)$.
\begin{definition}
Given a continuous probability distribution $D$ and a parameter transformation
$T$, we define the operations: \label{def:Tdef1}
\begin{align*}
\mathsf{rand}(T(D)) & =\text{\ensuremath{\mathsf{T(rand}(D))}}\\
\mathsf{cdf}(T(D),x) & =\mathsf{cdf}(D,T^{-1}(x))\\
\mathsf{pdf}(T(D),x) & =\mathsf{pdf}(D,T^{-1}(x))\cdot(T^{-1})'(x)
\end{align*}
\end{definition}
\noindent This definition matches how probability distributions transform
in probability theory. Our implementation represents a parameter transformation
$T$ as the 4-tuple of functions $(T,T',T^{-1},(T^{-1})')$, so that
we have access to the inverse and derivative.
\begin{definition}
Given an interval $(a,w)$ with midpoint $a\in\mathbb{R}$ and width
$w\in\mathbb{R}$ , we let $l=T(a-\frac{w}{2})$ and $r=T(a+\frac{w}{2})$
and define: \label{def:Tdef2}
\begin{align*}
T((a,w)) & =\left(\frac{l+r}{2},r-l\right)
\end{align*}
\end{definition}
\noindent This performs parameter transformation on an interval represented
as a midpoint-width pair. If the width is infinitesimal, we need a
different rule.
\begin{definition}
Given an interval $(a,w)$ with midpoint $a\in\mathbb{R}$ and infinitesimal
width $w$, we define \label{def:Tdef3}:
\begin{align*}
T((a,w)) & =(T(a),T'(a)\cdot w)
\end{align*}
\end{definition}
\noindent This performs parameter transformation on an infinitesimal
width interval, which gets transformed to an interval whose width
is larger by a factor $T'(a)$. The key lemma about parameter transformations
is that they do not affect the value of the (possibly infinitesimal)
probability of a (possibly infinitesimal) interval.
\begin{lemma}
Let $T$ be a parameter transformation, $D$ a distribution, and $I$
an interval. Then $P(T(D),T(I))=P(D,I)$ where $P$ is the probability
function defined at (\ref{eq:P}).
\end{lemma}
\begin{proof}
We distinguish non-infinitesimal intervals from infinitesimal intervals.
\begin{itemize}
\item If $I=(a,w)$ is non infinitesimal, then by Definition (\ref{eq:P}):
\begin{align*}
P(D,(a,w)) & =\mathsf{cdf}(D,a+\frac{1}{2}w)-\mathsf{cdf}(D,a-\frac{1}{2}w)
\end{align*}
For $T((a,w))$ we have, where $l=T(a-\frac{w}{2})$ and $r=T(a+\frac{w}{2})$:
\begin{align*}
T((a,w)) & =(\frac{l+r}{2},r-l)
\end{align*}
and by (\ref{eq:P}):
\begin{align*}
P(T(D),T((a,w))) & =\mathsf{cdf}(T(D),\frac{l+r}{2}+\frac{1}{2}(r-l))-\mathsf{cdf}(T(D),\frac{l+r}{2}-\frac{1}{2}(r-l))\\
 & =\mathsf{cdf}(T(D),r)-\mathsf{cdf}(T(D),l)\\
 & =\mathsf{cdf}(D,T^{-1}(r))-\mathsf{cdf}(D,T^{-1}(l))\\
 & =\mathsf{cdf}(D,T^{-1}(T(a+\frac{w}{2})))-\mathsf{cdf}(D,T^{-1}(T(a-\frac{w}{2})))\\
 & \mathsf{=cdf}(D,a+\frac{w}{2})-\mathsf{cdf}(D,a-\frac{w}{2})
\end{align*}
\item If $I=(a,r\epsilon^{n})$ is infinitesimal ($n>0$), then by definition
(\ref{eq:P}):
\begin{align*}
P(D,(a,w)) & =\mathsf{pdf}(D,x)\cdot r\epsilon^{n}
\end{align*}
For $T((a,r\epsilon^{n}))$ we have:
\begin{align*}
T((a,r\epsilon^{n})) & =(T(a),T'(a)\cdot r\epsilon^{n})
\end{align*}
and by (\ref{eq:P}):
\begin{align*}
P(T(D),T((a,r\epsilon^{n}))) & =\mathsf{pdf}(T(D),T(a))\cdot T'(a)\cdot r\epsilon^{n}\\
 & =\mathsf{pdf}(D,T^{-1}(T(a)))\cdot(T^{-1})'(T(a))\cdot T'(a)\cdot r\epsilon^{n}\\
 & =\mathsf{pdf}(D,a)\cdot r\epsilon^{n}
\end{align*}
\end{itemize}
\end{proof}
This lemma implies that the effect of \code{observe(T(D),T(I))} is
the same as \code{observe(D,I)}, since \code{observe(D,I)} does
\code{weight {*}= P(D,I)}. This property of \code{observe} ensures
the absence of parameter transformation paradoxes, not only of the
three examples we gave, but in general: it does not matter which parameter
scale we use; the weight accumulated remains the same.

\section{Implementation in Julia\label{sec:IMPLEMENTATION-IN-JULIA}}

We have implemented the constructs described in the preceding sections
as a simple embedded DSL in the Julia programming language, with the
following interface:
\begin{itemize}
\item Infinitesimal numbers $r\epsilon^{n}$ constructed by \code{Infinitesimal(r,n)},
with predefined \code{eps = Infinitesimal(1.0,1)}, and overloaded
infinitesimal arithmetic operations \code{+,-,{*},/} according to
Definition \ref{def:infarith}.
\item Probability distributions \code{D} with random sampling \code{rand(D)}
and \code{cdf(D,x)} and \code{pdf(D,x)}. These distributions are
provided by Julia's Distributions package, which supports beta, normal,
Cauchy, Chi-square, Bernoulli, Binomial, and many other continuous
distributions and discrete distributions.
\item Intervals constructed by \code{Interval(mid,width)}, where \code{width}
may be infinitesimal, and an operation \code{P(D,I)} to compute the
(possibly infinitesimal) probability that a sample from $D$ lies
in the interval $I$. If $I$ is infinitesimal, then this uses the
PDF, and if $I$ has finite width, then this uses the CDF, according
to Definition \ref{Pdef}.
\item Parameter transformations \code{T} represented as 4-tuples $(T,T',T^{-1},(T^{-1})')$,
with operations \code{T(D)} and \code{T(I)} to transform probability
distributions and intervals, according to Definitions \ref{def:Tdef1},
\ref{def:Tdef2}, and \ref{def:Tdef3}.
\item The main operations of probabilistic programming DSL are the following:
\begin{itemize}
\item \code{rand(D)}, where \code{D} is a distribution provided by Julia's
Distributions package.
\item \code{observe(D,I)}, where \code{D} is a continuous distribution
and \code{I} is an interval, or \code{D} is a discrete distribution
and \code{I} is an element, implemented as \code{weight {*}= P(D,I)}
\item \code{importance(trials,program)} which does importance sampling,
where \code{trials} is the number of trials to run, and \code{program}
is a probabilistic program written as a Julia function that uses \code{rand}
and \code{observe}, and returns the value that we wish to estimate
the expectation value of. Importance sampling is implemented as described
in Section \ref{subsec:Importance-sampling-with}.
\end{itemize}
\end{itemize}
The example in the introduction can be written as follows:
\begin{verbatim}
  function example1_m()
     h = rand(Normal(1.7,0.5))
     if rand(Bernoulli(0.5))
        observe(Normal(2.0,0.1), Interval(h,eps))
     end
     return h
  end
  estimate = importance(1000000,example1_m)
\end{verbatim}
This program will produce an estimate very close to $1.7$. If we
change the units to centimeters, we will get an estimate very close
to $170$, as expected:
\begin{verbatim}
  function example1_cm()
     h = rand(Normal(170,50))
     if rand(Bernoulli(0.5))
        observe(Normal(200,10), Interval(h,100*eps))
     end
     return h
  end
  estimate = importance(1000000,example1_cm)
\end{verbatim}
The artifact contains the other examples from the paper and further
examples to illustrate the use of the DSL \citep{measurezeroartifact}.

\section{Conclusion \& Future Work}

We have seen that naive likelihood accumulation results in unit anomalies
when observe statements with continuous distributions are executed
conditionally on random data, and we have shown that the culprit is
the use of probability densities. From an analysis of what observe
statements mean in the discrete case, we motivated a switch to interval-based
observe statements, which have a probabilistic and rejection sampling
interpretation. To recover the behavior of measure-zero observe statements
we introduced intervals with infinitesimal width. This results in
the accumulation of infinitesimal probabilities rather than probability
densities, which solves the unit anomalies and paradoxes even when
conditioning on events of measure zero. Infinitesimal probabilities
also enabled us to implement parameter transformations that do not
change the behavior of the program. We implemented this form of probabilistic
programming as an embedded DSL in Julia.

This improves the state of the art in two ways:
\begin{enumerate}
\item It fixes unit and parameter transformation paradoxes, which result
in surprising and in some cases arguably incorrect behavior in existing
probabilistic programming languages when continuous observe statements
are executed conditionally on random data, or when nonlinear parameter
transformations are performed.
\item It gives the observe statement a probabilistic and\emph{ }rejection
sampling interpretation, with measure zero conditioning as a limiting
case when the observation interval is of infinitesimal width.
\end{enumerate}
We hope that this will have a positive impact on the development of
the formal semantic foundations of probabilistic programming languages,
potentially reducing the problem of conditioning to events of positive
measure. On the implementation side, we hope to generalize more powerful
inference algorithms such as Metropolis-Hastings and SMC to work with
infinitesimal probabilities.
\begin{acks}
I thank Sriram Sankaranarayanan and the anonymous reviewers for their
outstanding feedback. I am grateful to Arjen Rouvoet, Paolo Giarrusso,
Ike Mulder, Dongho Lee, Ahmad Salim Al-Sibahi, Sam Staton, Christian
Weilbach, Alex Lew, and Robbert Krebbers for help, inspiration, and
discussions.
\end{acks}

\bibliography{references}

\end{document}